\documentclass[pre,10pt,twocolumn]{revtex4}%
\usepackage{amsfonts}
\usepackage{amsmath}
\usepackage{amssymb}
\usepackage{graphicx}%
\newtheorem{theorem}{Theorem}

\newtheorem{corollary}[theorem]{Corollary}

\newtheorem{lemma}[theorem]{Lemma}

\newenvironment{proof}[1][Proof]{\noindent\textbf{#1.} }{\ \rule{0.5em}{0.5em}}
\begin{document}
\title{Control by quantum dynamics on graphs}
\author{Chris Godsil}
\affiliation{Department of Combinatorics \& Optimization, University of Waterloo, N2L 3G1
Waterloo, ON, Canada}
\author{Simone Severini}
\affiliation{Department of Physics and Astronomy, University College London, WC1E 6BT
London, United Kingdom}

\begin{abstract}
We address the study of controllability of a closed quantum system whose
dynamical Lie algebra is generated by adjacency matrices of graphs. We
characterize a large family of graphs that renders a system controllable. The
key property is a novel graph-theoretic feature consisting of a particularly
\emph{disordered} cycle structure. Disregarding efficiency of control
functions, but choosing subfamilies of sparse graphs, the results translate
into continuous-time quantum walks for universal computation.

\end{abstract}
\maketitle

\section{Introduction}

The study of classical control theory is ubiquitous across engineering
disciplines. The concept of controllability at the quantum level is a
fundamental notion which expresses the ability of implementing \emph{any}
dynamics in a given quantum mechanical set-up (see \cite{dalessandro} for an
introductory monograph). From the practical point of view, the successful
realization of quantum devices for a variety of information processing tasks
strongly depends on the ability of manipulating systems with sufficient
freedom. The design of protocols to control closed quantum systems mainly
deals with schemes for efficient controllability by acting on subspaces
\cite{daniel}.

A large number of systems have been shown to exhibit characteristics that
allow controllability. For example, almost any pair of Hamiltonians that can
be coherently applied to a finite-dimensional quantum system renders it
controllable, and almost any quantum logic gate is universal \cite{lloyd1}.
However, controllability of large systems does not directly give an efficient
implementation of control functions, and therefore processes like universal
quantum computation. Moreover, control criteria are generally not computable
for large systems and are not immediately scalable with the system size
\cite{crit}.

Here we consider controllability in relation to a class of dynamics that can
be interpreted as Schr\"{o}dinger evolutions of a particle hopping between the
vertices of a graph. The corresponding Hamiltonians are matrices with nonzero
entries only where transitions are permitted by the graph structure. This
class naturally embraces continuous-time quantum walks (for short, QWs)
\cite{walks}, and the evolution of a single (or multiple \cite{osb})
excitation subspace for systems of spin-half particles with various kind of
interaction (\emph{e.g.}, XY, XYZ, \emph{etc.}) \cite{state}. QWs have been
shown to build gates by scattering off a set of small graphs attached to wires
representing basis states \cite{childs}. QWs give examples of matrices for
sufficient control, when we can arbitrarily modify edge-weights (or,
equivalently, strength of couplings) \cite{carl}.

During our discussion, an $n$-level system is \emph{controllable} by a given
set of Hamiltonians (possibly acting on specific subsystems only) if every
element of the unitary group $U(n)$ can be approximated by the matrices of the
subgroup obtained via the dynamical Lie algebra of the set. This general
definition is useful to isolate the main difference between the notions of
controllability and universality. For instance, global evolution of a spin
system may require complex protocols to implement $2$-qubit gates on distant
sites, even if it permits complete controllability.

The specific problem addressed here is the following one: we study
controllability by the alternate application of two Hamiltonians. One of the
Hamiltonian describes a nearest-neighbour interaction defined by some graph.
The other Hamiltonian is a projector given by the characteristic vector of a
subset of vertices. This describes an interaction between every two spins
associated to the elements of the subset.

Our main result is to characterize a large family of graphs that give a pair
of Hamiltonians implementing any quantum dynamics, thereby rendering a system
controllable (Section II). The result can be seen as an analogue of the
Burgarth-Giovannetti infectivity criterion \cite{daniel1} in this setting
involving two Hamiltonians. A comparison with the infectivity criterion,
possibly with the use of the notion of zero-forcing \cite{zero}, remains an
open question. It is also an open problem to determine the necessity of our condition.

The proofs follow easily from facts of Lie theory and algebraic properties of
graphs. The members of our family present a particularly \emph{disordered
}cycle structure \cite{god}. Specifically we require that the number of cycles
of a certain length, starting from different vertices, cannot be written as a
sum of the numbers of smaller cycles. We will show that this is a property
responsible for controllability, when the Hamiltonian is the adjacency matrix
of the graph. Indeed, we will use powers of the adjacency matrix. These encode
the cycles of a graph (see the definition of a walk matrix below). In analogy
with a known result in quantum control theory of spin systems, the path graph
turns out to be the arguably most simple example \cite{dalessandro, crit}. The
setting is directly equivalent to a single excitation evolving on an XY spin
chain with constant couplings (here XY means $XX+YY$, \emph{\`{a} la} Bose
\cite{bose}). The path is a connected graph with minimal number of edges,
therefore it corresponds to a very sparse Hamiltonian. This is a fact to take
into account, since sparse Hamiltonians can be simulated efficiently in a
quantum computer \cite{be}.

In general, focusing only on the dynamics restricted to the $n$-dimensional
subspace, the physical device for implementation consists of any machine
realizing QWs (\emph{e.g.}, an optical waveguide lattice \cite{pe}). While we
do not modify the circuitry for different tasks, an external clock is
necessary, since we need to know the time of application of each Hamiltonian,
even if the resulting operation is just a phase factor. We will give evidence
that our property is almost sure (Section II). This is parallel to the fact
that almost every (generic) Hamiltonian gives sufficient control. For many
types of graphs, the construction of an infinite family with a typical
property may not be straightforward (\emph{e.g.}, expanders, small diameter
graphs, \emph{etc.}). We will present a method to construct infinite families
of graphs (without special constraints) that satisfy our required property
(Section III).

The results of the paper may translate into valuable information in the
perspective of designing schemes for scalable quantum computation via local
control (\emph{e.g.}, help in the selection of the systems, the engineering of
control functions, \emph{etc.}). (See \cite{daniel, crit} for extended
treatments of this topic.) From a wider angle, the results consist of a step
towards a better and more general understanding quantum evolution on networks.
Additionally, we introduce concepts that propose an interface between control
theory and graph theory.

The paper has four further sections: Section II contains the general result;
Examples are in Section III; we draw some brief conclusions and state open
problems in Section IV.

\section{Controllability}

Let $X=(V(X),E(X))$ be a (simple) \emph{graph} with a set of $n$ vertices
$V(X)$ and a set of edges $E(X)\subseteq V(X)\times V(X)-\{\{i,i\}:i\in
V(X)\}$. The \emph{adjacency matrix} of $X$, denoted by $A(X)$, is an $n\times
n$ matrix with $A(X)_{i,j}=1$ if $\{i,j\}\in E(X)$ and $A(X)_{i,j}=0$,
otherwise. Let $X$ be a graph on $n$ vertices and let $z\in\mathbb{R}^{n}$. We
define and denote by $W_{z}(X)=\left(
\begin{array}
[c]{cccc}%
z & A(X)z & \ldots & A(X)^{n-1}z
\end{array}
\right)  $ an $n\times n$ matrix with entries in $\mathbb{Z}^{\geq0}$
associated to $X$. When $z$ is the characteristic vector of some set
$S\subseteq V(X)$, the matrix $W_{z}(X)$ is called a \emph{walk matrix} of $X$
with respect to $S$. In this case, we may write $W_{S}(X)$ instead of
$W_{z}(X)$. The pair $\left(  X,z\right)  $ is said to be \emph{controllable}
if the matrix $W_{z}(X)$ is invertible (\emph{i.e.}, $\det(W_{z}(X))\neq0$).
When $z$ is the characteristic vector of some set $S\subseteq V(X)$, we may
write $\left(  X,S\right)  $ instead of $\left(  X,z\right)  $. The walk
matrix of a graph contains information about its cycle structure \cite{ha}.
The definition of a controllable graph arises as a special case of a
controllable pair. Let $X$ be a graph and let $\mathbf{1}$ be the all-ones
vector. This is also the characteristic vector of $V(X)$. The graph $X$ is
said to be \emph{controllable} if $\left(  X,\mathbf{1}\right)  $ (or,
equivalently, $\left(  X,V(X)\right)  $) is controllable.

Let us recall that a \emph{walk of length} $l$ in a graph $X$ is a sequence of
vertices $1,2,...,l,l+1$, such that $\{i,i+1\}\in E(X)$, for every $1\leq
i\leq l$. The $ij$-th entry of the walk matrix, $[W_{\mathbf{1}}%
(X)]_{i,j}=\sum_{j=1}^{n}A_{i,j}^{l-1}(X)$, counts the number of all walks of
length $l-1$ from vertex $i$. Let $d(i):=|\{j:\{i,j\}\in E(X)\}|$ be the
\emph{degree }of a vertex $i$. A graph $X$ is \emph{regular} if $d(i)$ is
constant over $V(X)$. Notice that a controllable graph can not be regular. In
fact, the walk matrix of a regular graphs has rank $1$, because $\mathbf{1}$
is one of its eigenvectors. One can verify by exhaustive search that there are
no controllable graphs on $n\leq5$ vertices. Fig. 1 below contains drawings of
all connected (non-isomorphic) controllable graphs on six vertices. Numerics
show that the ratio [number of graphs]/[number of controllable graphs]
decreases with $n$ (see caption of Fig. 1 for small examples). We expect that
asymptotically almost surely every graph is controllable, also considering
that the authomorphisms fixing the vertices of a controllable graphs are
trivial.
\begin{figure}
[h]
\begin{center}
\includegraphics[
height=1.0848in,
width=2.6769in
]%
{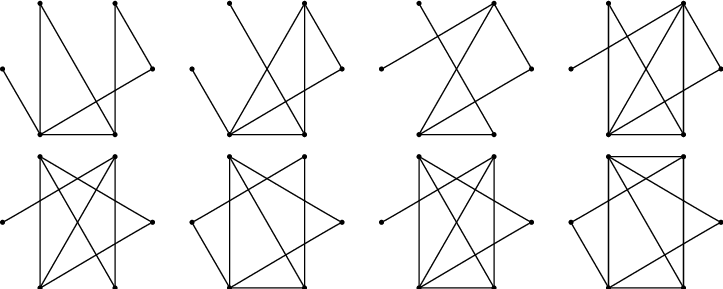}%
\caption{Drawings of all connected non-isomorphic controllable graphs on six
vertices. The vector $\left(  d(i\right)  :i\in V(X))$ is the \emph{degree
sequence} of $X$. The degree sequences of these graphs are particularly
\emph{irregular}. First line, from left to right:
(1,2,2,2,3,4),(1,1,2,3,3,4),(1,1,2,2,3,3), and (1,2,2,3,4,4). Second line,
from left to right: (1,2,2,3,3,3), (2,2,2,3,3,4), (1,2,3,3,3,4), and
(2,2,3,3,4,4). There are exactly $8,85,2275$ (connected) controllable graphs
on $6,7$, and $8$ vertices, respectively.}%
\end{center}
\end{figure}

\bigskip

A \emph{(continuous-time) quantum walk} on a graph $X$, starting from a state
$|\psi_{0}\rangle\in\mathbb{C}^{n}$, is the process induced by the rule
$U_{M(X)}(t)|\psi_{0}\rangle\mapsto|\psi_{t}\rangle$, where $U_{M(X)}%
(t):=e^{-iM(X)t}$ ($t\in\mathbb{R}^{+}$) and $M(X)$ is a symmetric matrix with
nonzero entries corresponding to the edges of $X$ (\emph{e.g.}, adjacency
matrix, combinatorial Laplacian, \emph{etc.}). A probability distribution
supported by $V(X)$ is obtained by performing a projective measurement on the
state $|\psi_{t}\rangle$. The matrix $M(X)$ can also be seen as governing the
dynamics of a system of spin-half particles restricted to a single excitation
subspace. The dimension of such a subspace is in fact $n$. Here we work with
adjacency matrices only, but the results described are valid for any symmetric
matrix. Studies of perfect state transfer and entanglement transfer in spin
systems are often carried on with respect to this restriction \cite{state}.
QWs and their discrete analogues (\emph{e.g.}, coined quantum walks, scalar
quantum walks, \emph{etc.}) have found a number of algorithmic applications.
The reviews \cite{walks} give a detailed perspective on this and related topics.

When choosing Hamiltonians of the form of $M(X)$, the question to ask about
controllability is the following one: can we obtain any quantum dynamics on an
$n$-level system by performing repeated applications of QWs? Moreover, how
much can we limit our resources (\emph{e.g.}, number of non-null interactions,
number of different Hamiltonians, \emph{etc.})? In particular, can we use just
a single QW (\emph{i.e.}, a fixed graph) plus an extra operation acting on a
subspace of a relatively small dimension? The latter one is linked to core
questions in quantum control theory, where we are interested in driving global
dynamics by directly modifying only a limited portion of the system under a
parsimony criterion. In the quantum mechanical set-up, controllability occurs
together with the ability of constructing with reasonable accuracy any unitary
matrix of the appropriate dimension (see Chapter 3 of \cite{dalessandro}). The
corresponding property is expressed if we guarantee density in $U(n)$ of the
group of unitaries realized as sequences of QWs. The next technical lemma
describes a relation between controllable pairs and Lie algebras.

\begin{lemma}
\label{lem}Let $X$ be a graph and let $z$ be the characteristic vector of a
set $S\subseteq V(X)$. Let us define the symmetric $(0,1)$-matrix $L=zz^{T}$.
If $(X,S)$ is a controllable pair then the real Lie algebra generated by the
matrices $A(X)$ and $L$ is \emph{Mat}$_{n\times n}(\mathbb{R})$, the algebra
of all $n\times n$ real matrices. The real Lie algebra generated by $iA(X)$
and $iL$ is the vector space of all skew-Hermitian matrices.
\end{lemma}

\begin{proof}
We prove by induction on $k$ that the Lie algebra generated by $A=A(X)$ and
$L$ contains the matrices $A^{k-i}LA^{i}$, with $i=0,\ldots,k$. The first
claim in the lemma will follow at once from this. We note that there are
integers $c_{r}$ such that $LA^{r}L=c_{r}L$. If our Lie algebra contains the
matrices $A^{k-i}LA^{i}$, then it contains the Lie products $A^{k+1-i}%
LA^{i}-A^{k-i}LA^{i+1}$, with $i=0,\ldots,k$, and the partial sums
$A^{k+1-i}LA^{i}-LA^{k+1}$, for all $i$. In particular, it contains
$A^{k+1}L-LA^{k+1-i}$ and therefore also
\begin{align*}
&  LA^{k+1}L-L^{2}A^{k+1}-A^{k+1}L^{2}+LA^{k+1}L\\
&  =2c_{k+1}L-c_{0}(A^{k+1}L+LA^{k+1}).
\end{align*}
for all $i$. From this, it follows that it contains $LA^{k+1}$, and therefore
all the monomials $A^{k+1-i}LA^{i}$. Let us now consider the second claim of
the lemma. We say a matrix is a \emph{commutator of degree} $r+1$, if it can
be written as $AX-XA$ or $LX-XL$ for some commutator of degree $r$, where the
commutators of degree zero are the matrices in the span of $A$ and $L$. Since
$A$ and $L$ are symmetric, we see that all commutators of even weight are
symmetric and all commutators of odd weight are skew-symmetric. The
intersection of the space of symmetric matrices with the space of
skew-symmetric matrices is the zero subspace, from which we deduce that the
even-weight commutators span the space of real symmetric matrices and the
odd-weight commutators span the space of skew-symmetric matrices. This implies
that the even-weight commutators in $iA$ and $iL$ span the space of skew
symmetric matrices, with dimension $(n^{2}-n)/2$, while the odd-weight
commutators span a complementary space of dimension $(n^{2}+n)/2$. This proves
our second claim.
\end{proof}

\bigskip

We remark that it is not hard to show that the matrices
\[%
\begin{tabular}
[c]{lll}%
$\left(
\begin{array}
[c]{cc}%
0 & 1\\
0 & 0
\end{array}
\right)  $ & and & $\left(
\begin{array}
[c]{cc}%
0 & 0\\
1 & 0
\end{array}
\right)  $%
\end{tabular}
\ \
\]
generate \emph{Mat}$_{2\times2}(\mathbb{R})$, but the real Lie algebra they
generate is $\mathfrak{s\ell}(2,\mathbb{R})$ rather than $\mathfrak{g\ell
}(2,\mathbb{R})$. Lemma \ref{lem} is reminiscent of the Lie Algebra Rank
Condition in quantum control theory. The following result gives a sufficient
condition to render a system controllable by a QW:

\begin{theorem}
\label{thm:}Let $X$ be a graph and let $z$ be the characteristic vector of a
set $S\subseteq V(X)$. If $(X,S)$ is a controllable pair then the unitary
matrices $U_{A(X)}(s)=e^{-iA(X)s}$ and $U_{L}(t)=e^{-iLt}$, $s,t\in
\mathbb{R}^{\geq0}$, generate a dense subgroup of the unitary group $U(n)$,
$n\geq2$.
\end{theorem}

\begin{proof}
Let $G$ be the closed subgroup generated by the given elements. Then it is a
Lie subgroup of $U(n)$, and its tangent space is the Lie algebra generated by
$iA(X)$ and $iL$. Since $X$ is controllable, by Lemma \ref{lem}, this Lie
algebra is the space of all skew-Hermitian matrices, which is the tangent
space to the unitary group $U(n)$. It follows that $G=U(n)$.
\end{proof}

\bigskip

If $S=V(X)$ then $z=\mathbf{1}$ and the Hamiltonian $L=J$, the all-ones
matrix. For its unitary, we have $U_{J}(t)_{k,l}=\frac{1}{n}\left(
n+e^{-int}-1\right)  $ if $k=l$ and $U_{J}(t)_{k,l}=\frac{1}{n}\left(
e^{-nit}-1\right)  $, otherwise. In fact $U_{M}(t)$ is a polynomial in $M$
with degree at most the degree of the minimal polynomial of $M$. Among
adjacency matrices, this is the \emph{fullest}\ possible Hamiltonian. Its
implementation has been discussed in several works \cite{du}. The matrix
$U_{J}(t)$ is essentially the same as the Grover operator used in quantum
search algorithms \cite{ki}. Turning our attention to different characteristic
vectors, we can prove a similar result concerning the path graph. We denote by
$P_{n}$ the \emph{path} of length $n-1$, \emph{i.e.}, the graph on $n$
vertices $\{1,2,...,n\}$ and edges $\{1,2\},\{2,3\},...,\{n-1,n\}$ ($1$ and
$n$ are called \emph{end-vertices}). Weighted paths are often used to model 1D
spin chains. A\ control criterion concerning the global space of 1D spin
chains have been isolated in \cite{crit}. Controllability of these systems and
scalable quantum computation has also been discussed by in \cite{crit}.

\begin{corollary}
\label{cor}Let $P_{n}$ be the path on $n$ vertices. The unitary matrix
$U_{A(P_{n})}(s)=e^{-iA(P_{n})s}$ together with the diagonal unitary matrix
$U_{e_{1}e_{1}^{T}}(t)=e^{-ie_{1}e_{1}^{T}t}=[e^{-it},1,...,1]$,
$s,t\in\mathbb{R}^{\geq0}$, where $e_{1}=(1,0,...,0)^{T}$, generate a dense
subgroup of the unitary group $U(n)$, with $n\geq2$.
\end{corollary}

\begin{proof}
We need to prove that $W(P_{n})=\left(
\begin{array}
[c]{cccc}%
e_{1} & A(P_{n})e_{1} & \ldots & A(P_{n})^{n-1}e_{1}%
\end{array}
\right)  $ is invertible. Observe that the first entry of the vector
$A^{l}(P_{n})e_{1}$ is a Catalan number $C_{l/2}=\frac{2}{l+2}\binom{l}{l/2}$
if $l$ is even and zero, otherwise; the second entry behaves similarly, with
$C_{(l+1)/2}=\frac{2}{l+3}\binom{l+1}{\left(  l+1\right)  /2}$, but for $l$
odd \cite{sp}. For example, the first two rows of $W(P_{7})$ are $\left(
1,0,1,0,2,0,5\right)  $ and $\left(  0,1,0,2,0,5,0\right)  $. Moreover, the
matrix $W(P_{n})$ is upper triangular. Since $C_{k}\succ\sum_{i=0}^{k}C_{i}$,
for every $k\geq3$, it follows that the rows of $W\left(  P_{n}\right)  $ are
linearly independent.
\end{proof}

\section{Examples}

Given $S\subseteq V(X)$, the \emph{cone} of $X$ relative to $S$ is the graph
$\widehat{X}_{S}$ such that $V(\widehat{X}_{S})=V(X)\cup\{0\}$, for a new
vertex $0$, and $E(\widehat{X}_{S})=E(X)\cup\{\{0,i\}:i\in S\}$. We denote by
$X\backslash i$ the graph obtained from $X$ by deleting the vertex $i$ and all
its incident edges.

\begin{theorem}
Given a graph $X$ and a vertex $1\in V(X)$, the pair $(\widehat{X}_{1},\{0\})$
is controllable if $(X,\{1\})$ is.
\end{theorem}

\begin{proof}
We will show that if $u\in V(Z)$ for some graph $Z$, then $(Z,\{u\})$ is
controllable if and only if the characteristic polynomials (of the adjacency
matrices) $\phi(Z,t)$ and $\phi(Z\backslash u,t)$ are coprime. From the
properties of $\widehat{X}_{1}$, one can prove that $\phi(\widehat{X}%
_{1},t)=t\phi(X,t)-\phi(X\backslash1,t)$. From this, we deduce that if
$\phi(X,t)$ and $\phi(X\backslash1,t)$ are coprime then so are $\phi
(\widehat{X}_{1},t)$ and $\phi(X,t)$. Now we derive our characterization of
controllability. Assume $n=|V(X)|$. Let $e_{1}$ be the first vector of the
standard basis, and let $E_{\theta}$ denote the idempotent in the spectral
decomposition of $A=A(X)$ that corresponds to $\theta$ (see \cite{gr}, pp.
186--187) it follows that%
\[
\frac{\phi(X\backslash1,t)}{\phi(X,t)}=[(tI-A)^{-1}]_{1,1}=\sum_{\theta
}(t-\theta)^{-1}e_{1}^{T}E_{\theta}e_{1}.
\]
We observe that the number of poles in the rational function here is equal to
the number of eigenvalues $\theta$ such that $e_{1}^{T}E_{\theta}e_{1}\neq0$;
in other words, it is equal to the number of eigenvalues $\theta$ such that
the projection $E_{\theta}e_{1}\neq0$. Note also that this number is $n$ if
and only if $\phi(X,t)$ and $\phi(X\backslash1,t)$ are coprime. To complete
the argument, consider the walk matrix $W_{e_{1}}(X)$. By spectral
decomposition (again) $A^{r}e_{1}=\sum_{\theta}E_{\theta}e_{1}$, from which it
follows that the column space of $W_{e_{1}}(X)$ lies in the span of the
nonzero vectors $E_{\theta}e_{1}$. Since each projection $E_{\theta}$ is a
polynomial in $A$, we conclude that rk$(W)$ is equal to the number of
eigenvalues $\theta$ of $X$ such that $E_{\theta}e_{1}\neq0$. This proves our characterization.
\end{proof}

\bigskip

A method based on the theorem can be used to construct infinite families of
controllable graphs:

\begin{corollary}
Let $(X,S)$ be a controllable pair. If $Y$ is the graph obtained by joining
one end-vertex of the path to each vertex in $S$, then $Y$ is controllable.
\end{corollary}

\section{Conclusions}

We have considered controllability and QWs. As a technical tool, we have
introduced the combinatorial notion of a controllable pair. A graph and a
subset of its vertices form a controllable pair, when the structure of the
graph exhibit a certain type of disorder. The disorder is expressed by the
cycle structure of the graph, encoded in the entries of powers of the
adjacency matrix. We have proved that a QW involving a such a pair renders a
system controllable. By this result, we can \emph{in principle }perform
universal quantum computation as an alternating sequence of QWs on two graphs,
or on the same graph, but interspersed with phase factors. Including fault
tolerance in this picture would encounter hard obstacles because of the
sensitivity to phenomena linked to decoherence and Anderson localization
\cite{k}. An issue related to the more abstract aspects is the lack of
transparency when trying to design algorithms with a logic that requires
operations on specific subsystems. \ We conclude by stating four problems:
\emph{Problem 1: }Let $G$ be a subgroup of $U(n)$ which fixes $|\psi\rangle
\in\mathbb{C}^{n}$ ($n\geq3$) and let $V\in U(n)-G$. Then, the group
$\left\langle G,V\right\rangle $ (\emph{i.e.}, the subgroup of $U(n)$
generated by $G$ and $V$) is dense in $U(n)$ (see, \emph{e.g.}, Lemma 20 in
\cite{ah}). Is there an analogue to this statement for dynamical Lie algebras
generated by adjacency matrices? In particular, let $z$ be the characteristic
vector of $S\subseteq V(X)$ and let $P$ be a permutation matrix corresponding
to an automorphism of $X$. When $Pz=z$, it follows that $PW_{S}(X)=W_{S}(X)$
and then $P=I$ because $\det(W_{S}(X))\neq0$. This means that the
authomorphisms of $X$ fixing $S$ are trivial if $\left(  X,S\right)  $ is
controllable. Is this the most general condition for controllability?
\emph{Problem 2: }Can we lift the combinatorial criterion for controllability
introduced in this paper to general criteria for controllability of spin
systems? \emph{Problem 3:} Study controllability by adjacency matrices, when
the time of application of each Hamiltonian is constrained. Determining
relations between quantum control by nearest-neighbour interaction on graphs
and classical simulatability of the associated dynamics is an open problem.
\emph{Problem 4. }What can be said about controllability by acting only on a
connected induced subgraph? If the number of vertices is constrained, the
optimum may be difficult to compute. This would be parallel to the
Burgarth-Giovannetti criterion whose optimum is difficult even to approximate
\cite{aazami}.

\emph{Acknowledgments. }The authors would like to thank Daniel Burgarth and
Alastair Kay for very valuable conversation about quantum control theory and
for reading earlier drafts; Domenico D'Alessandro for several important
remarks; Alessandro Cosentino for drawing the graphs in Fig. 1; the anonymous
referees for suggestions that helped to improve the paper. SS is supported by
a Newton International Fellowship.

\end{document}